\documentclass[12pt]{article}
\usepackage[margin=0.9in]{geometry}

\usepackage[utf8]{inputenc}
\usepackage{bm}
\usepackage{bbm}
\usepackage{stmaryrd}
\usepackage{amsmath}
\usepackage{amsthm}
\usepackage{amssymb}
\usepackage{xcolor}
\usepackage{graphicx}
\usepackage{ulem}
\usepackage{url}
\usepackage{comment}
\usepackage{braket}

\DeclareFontFamily{U}{mathx}{}
\DeclareFontShape{U}{mathx}{m}{n}{<-> mathx10}{}
\DeclareSymbolFont{mathx}{U}{mathx}{m}{n}
\DeclareMathAccent{\widehat}{0}{mathx}{"70}
\DeclareMathAccent{\widecheck}{0}{mathx}{"71}

\newcommand{\tr}{\mathrm{Tr}}

\newcommand{\h}{{\mathcal H}}

\renewcommand{\r}{{\rm R}}

\newcommand{\s}{{\rm S}}
\newcommand{\z}{{\rm Z}}
\newcommand{\cx}{{\mathbb C}}
\newcommand{\rx}{{\mathbb R}}

\newcommand{\ve}{{\varepsilon}}

\newcommand{\bbbone}{\mathchoice {\rm 1\mskip-4mu l} {\rm 1\mskip-4mu l}
{\rm 1\mskip-4.5mu l} {\rm 1\mskip-5mu l}}
\newtheorem{thm}{Theorem}
\newtheorem{prop}{Proposition}
\newtheorem*{prop*}{Proposition}

\newtheorem{cor}{Corollary}

\newtheorem*{definition}{Definition}

\newcommand{\llim}{\lim_{\substack{t\rightarrow 0_+,\,\lambda\rightarrow\infty\\
\lambda t =  \tau \ \rm fixed}}}

\usepackage{authblk}
\title{Ultrastrongly coupled open systems \\
and fine grained time
}
\author[1,2]{Stefano Marcantoni\footnote{stefano.marcantoni@gssi.it}}
\author[3]{Marco Merkli\footnote{merkli@mun.ca}}
\affil[1]{Laboratoire J.~A.~Dieudonn\'e (LJAD) 

Universit\'e C\^ote d'Azur

Parc Valrose, 06108 Nice, France
\medskip
}
\affil[2]{Mathematics Division 

Gran Sasso Science Institute

Viale Rendina 24-26-28, 67100 L'Aquila, Italy
\medskip
}
\affil[3]{Department of Mathematics and Statistics

Memorial University of Newfoundland

St.~John’s, NL, Canada A1C 5S7
}

\begin{document}

\maketitle

\begin{abstract} 
We study the dynamics of a $d$-level quantum system coupled to a bosonic reservoir when the coupling constant is large. It is known that in the limit of infinite coupling strength, the system undergoes an instantaneous nonselective measurement, resulting in the immediate decoherence in the measurement basis, followed by  a unitary Zeno dynamics. Here we resolve this dynamical process by introducing a fine grained scaling regime of short times proportional to the inverse coupling. We provide a rigorous derivation of the open system dynamics in this regime of ultrastrong coupling and demonstrate how decoherence unfolds continuously in the new time scale. We show that Markovian dynamics which are not given by semigroups arise naturally, in contrast to what happens in the weak coupling theory. 
\end{abstract}

\section{Introduction}

An open quantum system is a (typically small) quantum system $\s$ interacting with another system (much larger and typically infinite)  called the reservoir $\r$. The compound $\s\r$ evolves according to the Schr\"odinger-von Neumann equation generated by a total, interacting Hamiltonian and one is interested in the reduced dynamics of $\s$, describing the evolution of observables (or that of the reduced density matrix) pertaining to $\s$ alone. This dynamics is in general dissipative and non-Markovian, and in all but very special cases, the exact form is too complicated to be determined. For this reason, the derivation of good approximations in suitable coupling regimes and time-scales is a topic of intensive research. 

Well known results have been obtained in the seventies \cite{Da74,Da76} for the so-called {\it van-Hove} or {\it ultraweak coupling regime}, where a vanishingly small coupling constant $\lambda \to 0$ and a long-time scale $t \to \infty$ are considered such that $\tau' = \lambda^2 t$ is finite \cite{VH55} (see Figure \ref{Fig:Regime}). In this limit, the system dynamics is given by a Markovian semigroup whose generator has a specific, so-called  GKSL structure (Gorini-Kossakowski-Sudarshan-Lindblad) \cite{GKS76,Lindblad,ChrPas}. The dynamical equation for the system density matrix is the (Markovian) master equation. The literature on the ultraweak coupling regime is very large and still developing | refined approximations in the same scaling regime were obtained more recently in \cite{Rivas10,Rivas17}. When the coupling constant $\lambda$ is small but fixed, independent of time $t\ge 0$, then one talks of the {\it weak coupling regime}. Showing that the master equation is valid in this regime for all times $t\ge 0$ | that is with error bounds small in $\lambda$ and uniform in time $t\ge0$ | is technically somewhat more involved. This has been done by means of the quantum resonance theory \cite{MSB2008,MMAOP,MMQI,MMQII}. In \cite{MM22} the weak coupling theory is extended to initially correlated $\s\r$ states and to the dynamics of reservoir observables. It is proved there that the Born and the Markov approximations are valid uniformly in time $t\ge 0$ (even as $t\rightarrow\infty$) provided $\lambda$ is small, fixed.

\begin{figure}[h!]
\centering
\includegraphics[width=.96\textwidth]{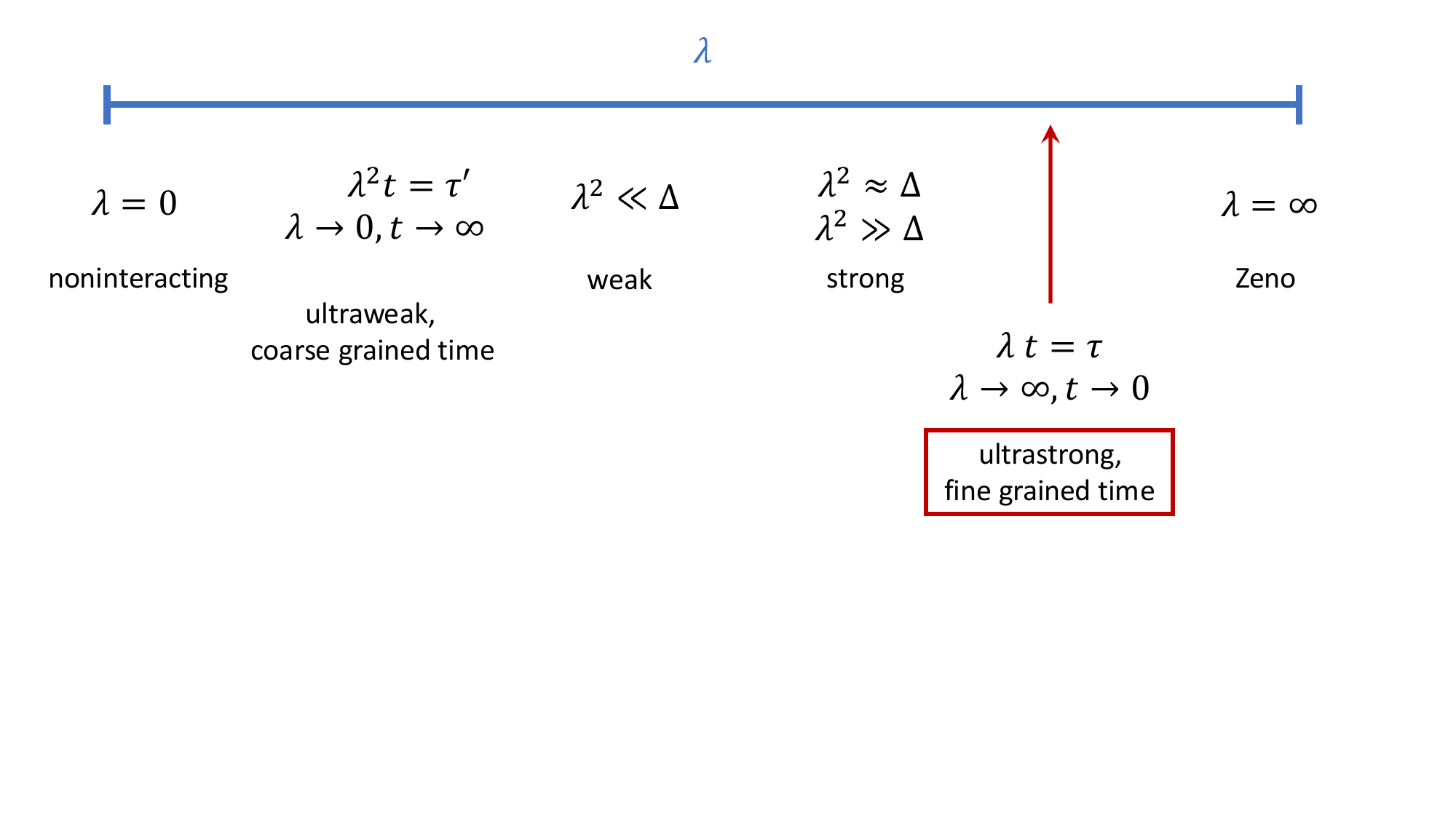}
\vspace*{-3.0cm}
\caption{Different coupling regimes: From noninteracting to Zeno. Here, $\lambda$ is the coupling parameter and $\Delta$ represents a typical system Bohr frequency, a hopping matrix element in the system Hamiltonian, or similar. The ultraweak and ultrastrong coupling regimes are {\it dynamical} regimes involving rescaled times. Here we introduce the fine grained time $\tau$ for the ultrastrong coupling regime and we derive the system dynamics as a function of $\tau$. We show that the instantaneous decoherence process known to happen in the Zeno limit (`$\lambda=\infty$') is resolved as a continuous dynamical process in the fine grained time. }
\label{Fig:Regime} 
\end{figure}

The case of strong coupling is characterized by $\lambda$ exceeding typical system parameters (Bohr frequencies, hopping terms). It has been investigated more recently. If the system Hamiltonian commutes with the interaction operator then the reduced dynamics is exactly solvable for arbitrary values of the coupling constant $\lambda$, by a so-called polaron transformation. This gives rise to the strategy of splitting the dynamics into an exactly solvable part for arbitrary $\lambda$ and times $t$, which is then perturbed by the non commuting part of the system Hamiltonian. Rigorous results in this regime have been obtained in \cite{M14,M16}; see also the recent analysis of \cite{Tru22,KhanAgar}. A different approach is the so-called reaction-coordinate mapping, that amounts to incorporating some degrees of freedom of the reservoir into the system and treating the remaining ones perturbatively as an effective bath \cite{Segal23a,Segal23b,Stras16}. Despite these advances, so far the methods are tailored to specific models and a general mathematical treatment is not available for the regime of strong coupling \cite{TMCA22}, in contrast to the ultraweak and weak coupling regimes. 

The {\it Zeno coupling} limit is $t>0$ fixed and $\lambda \to \infty$. It is shown in \cite{Marcantoni-Merkli} that the system undergoes the {\it Zeno effect} in this limit, namely, it experiences instantaneous nonselective measurement in the eigenbasis of the coupling operator, followed by a residual dynamics generated by a Zeno Hamiltonian. The effect of the nonselective measurement on the system density matrix, is to delete (set to zero) all non-diagonal density matrix elements while leaving the diagonal unchanged. This means that the system undergoes instantaneous decoherence in the given basis due to the Zeno, ``$\lambda=\infty$'' coupling limit. This causes a discontinuity in the system state at time $t=0_+$. Of course one expects that if $\lambda$ is large but finite, then the system decoherence due to the coupling to the reservoir, would proceed at a finite speed and in a continuous fashion. Therefore, in an attempt to temporally resolve the decoherence process, we investigate here the regime of large but finite $\lambda$, and small but nonzero times $t$, such that $\tau=\lambda t$ is an arbitrary but fixed number,
\begin{equation}
\label{regime0}
\lambda\rightarrow\infty,\quad t\rightarrow 0_+,\qquad \tau=\lambda t \in (0,\infty).
\end{equation}
We call \eqref{regime0} the ultrastrong coupling regime\footnote{In \cite{Marcantoni-Merkli}  ultrastrong coupling refers to the limit $\lambda \to \infty$ with fixed $t$. We propose to call this the Zeno regime while ultrastrong is reserved for \eqref{regime0}. This seems more natural by analogy with the ultraweak coupling regime.} (see also Figure \ref{Fig:Regime}) and $\tau$ is the {\it fine grained} time. It is akin to the coarse grained $\tau'=\lambda^2 t$ with $\lambda\rightarrow 0$ and $t\rightarrow\infty$ appearing in the ultraweak coupling theory.

We show in Section \ref{sec:finegraining} that in the ultrastrong coupling limit \eqref{regime0}, the system density matrix $\varrho_\s(\tau)$ varies continuously in the fine grained time $\tau$, and connects continuously to the initial system state for $\tau\rightarrow 0_+$ and to the fully decohered (measured) density matrix at $\tau\rightarrow\infty$. The decoherence process is determined by a decoherence function which depends on the properties of the reservoir initial state and the reservoir part of the interaction. For a typical $\s\r$ model with a bosonic reservoir (of the `spin-boson' type) in a thermal equilibrium state (at any temperature | or more generally for any Gaussian reservoir state), the decoherence function behaves as $e^{-\alpha \tau^2}$ for some $\alpha>0$. 

In Section \ref{sec:markov} we analyze the Markovianity properties of the dynamics $\varrho_\s(\tau)=\Lambda(\tau)\rho_\s(0)$. We show in particular that for a large group of models (including all Gaussian reservoir states) the dynamics is  Markovian, but it is not given by a semigroup. The introduction of the fine grained time scaling regime and the analysis of the resulting open system evolution, is the main novelty of this work. Here we deal with finite dimensional systems $\s$. The question of decoherence in continuous variable systems, in particular spatial decoherence, is of great interest in the context of open many-body systems. We study this in \cite{Mar-Mer-3}. 
\medskip

{\it Structure of the paper.}  In Section~\ref{sec:zenomeasuproc} we recall the setting and main result of \cite{Marcantoni-Merkli} on the Zeno dynamics in the limit of infinite coupling. In Section~\ref{sec:finegraining} we introduce the new scaling (fine graining of time) that allows to resolve the decoherence process. We derive the open system dynamics in this ultrastrong coupling regime for arbitrary coupling operators. In Section~\ref{sec:markov} we discuss the characterization of Markovianity for the mentioned dynamics, building on results already available in the literature for pure dephasing processes. We analyze the case of spin-boson models in Section~\ref{sec:spinboson}.

\section{Zeno effect for spin-boson systems as $\lambda \to \infty$}
\label{sec:zenomeasuproc}

In this section we present the model of an open system which was used in \cite{Marcantoni-Merkli} to derive the Zeno measurement effect on a system $\s$ caused by infinitely large coupling to an environment $\r$. The main result of \cite{Marcantoni-Merkli} is Theorem \ref{thm1} below.

A $d$-level quantum system is coupled to a reservoir with a continuum of modes labeled by $k\in\mathbb R^3$, with  Bosonic creation and annihilation operators $a^\dag(k),a(k)$, satisfying the canonical commutation relation $[a(k),a^\dag(l)]=\delta(k-l)$. The total Hamiltonian is given by
\begin{equation}
\label{1}
H=H_\s +H_\r +\lambda G\otimes \varphi(g),
\end{equation}
where $H_\s$ is the system Hamiltonian, that is a $d\times d$ hermitian matrix and 
\begin{equation}
\label{2}
H_\r = \int_{\mathbb R^3} \omega(k) a^\dag(k)a(k) d^3k
\end{equation}
is the reservoir Hamiltonian. In \eqref{2}, $\omega(k)\ge 0$ is the energy of the mode $k$ (`dispersion relation'). One may think of the photon dispersion relation $\omega(k) = |k|$, but this is not necessary for our analysis. The interaction term in \eqref{1} carries a {\it coupling constant} $\lambda\in\mathbb R$, a system coupling operator 
\begin{equation}
\label{3.0}
G = \sum_{l=1}^\nu \gamma_l P_l
\end{equation}
with distinct (possibly degenerate) eigenvalues $\gamma_j$ and spectral projections $P_l$ ($\dim P_l\ge 1$), and the field operator 
\begin{equation}
\label{3}
\varphi(g) = \frac{1}{\sqrt 2} \int_{\mathbb R^3} \big( g(k) a^\dag(k) + {\rm h.c.} \big)d^3k,
\end{equation}
where $g(k)$ is a complex valued function, called the form factor.

As is known (and discussed in \cite{Marcantoni-Merkli}), because $H_\r$  has absolutely continuous spectrum one cannot define the reservoir equilibrium state as a density matrix $\propto e^{-\beta H_\r}$ (this operator is not trace class, ${\rm tr} e^{-\beta H_\r}=\infty$). 
Rather, the construction of the continuous mode equilibrium state is done by taking a limit of discrete mode equilibrium states (the `thermodynamic limit'). It results in an expectation functional $\omega_{\r,\beta}$ for reservoir observables (built from functions of $a^\dag(k)$, $a(l)$), which can be expressed entirely by the {\it characteristic function}
\begin{equation}
\label{5}
\omega_{\r,\beta}(W(f)) =e^{-\frac14 \langle f,\coth(\beta\omega/2)f\rangle}.
\end{equation}
Here, $\langle f,h\rangle = \int_{\mathbb R^3} \overline{f(k)} g(k) d^3k$ is the inner product of $L^2({\mathbb R^3}, d^3k)$ and $W(f)$ is the unitary {\it Weyl operator},
\begin{equation}
\label{6}
W(f) =e^{i\varphi(f)}
\end{equation}
with $\varphi(f)$ as in \eqref{3}. The characteristic function \eqref{5} is also called the generating function, as it can be used to express the expectation for any observable by using the relation $\varphi(f) = -i\partial_\alpha|_{\alpha =0} W(\alpha f)$. One then finds that the two-point function of the reservoir equilibrium state is given by
\begin{equation}
\label{7}
\omega_{\r,\beta}\big(a^\dag(k)a(l)\big) = \frac{\delta(k-l)}{e^{\beta\omega(k)}-1}.
\end{equation}
This encodes Planck's law of black body radiation, where $n(k) = (e^{\beta\omega(k)}-1)^{-1}$ is the momentum density distribution in the reservoir. (That is, the number of modes per unit volume in a given region $\Lambda\in\mathbb R^3$ in momentum space is $\int_{\Lambda} n(k) d^3k$.) 

The state \eqref{5} is Gaussian and centered. Its covariance operator $\mathcal C$ acting on $L^2(\mathbb R^3,d^3k)$ is the operator of multiplication with the function $\coth(\beta\omega/2)$. We may consider more general Gaussian states of the form
\begin{equation}
\label{9}
\omega(W(f)) = e^{-\frac14 \langle f , \mathcal C f\rangle},
\end{equation}
where $\mathcal C$ is a general covariance operator on $L^2(\mathbb R^3,d^3k)$, satisfying
\begin{equation}
\label{C>1}
\mathcal C\ge \bbbone.
\end{equation}
The condition \eqref{C>1} is known to be necessary and sufficient for the right hand side of \eqref{9} to be the expectation functional of a quantum state; the case $\mathcal C=\bbbone$ is the field vacuum (zero temperature case).  Instead of the thermal distribution \eqref{7} we may consider reservoir states with an arbitrary energy distribution $\mu(k)\ge 0$, $\omega_\r\big(a^\dag(k)a(l)\big) = \mu(k) \delta(k-l)$, 
which corresponds to the covariance (compare with \eqref{5}) $\mathcal C  = C(k) = 1+2\mu(k)$. The corresponding state $\omega_\r$ \eqref{9} is stationary, $\omega_\r(e^{it H_\r}W(f)e^{-it H_\r}) = \omega_\r(W(e^{it\omega}f))=\omega_\r(W(f))$. Covariance operators $\mathcal C$ which are not multiplication operators by a function of $k$ result in non-stationary  Gaussian reservoir states, and they are included in the discussion and results. 
\medskip

We take initial system-reservoir states are of the form 
\begin{equation}
\label{10}
\omega_{\s\r} = \omega_\s\otimes\omega_\r,
\end{equation}  
where $\omega_\r$ is the Gaussian state \eqref{9} for a general covariance operator $\mathcal C\ge \bbbone$, and where \begin{equation}
\label{11}
\omega_\s (\cdot) = {\rm tr}_\s \big(\rho_\s \cdot \big)
\end{equation}
is a system state determined by a density matrix $\rho_\s$ of the $d$-level system with Hilbert space $\mathbb C^d$. Let $A\in\mathcal B(\mathbb C^d)$ be a system observable. The reduced system density matrix $\rho_\s(t)$ at time $t\ge 0$ in the Zeno limit is defined by the relation
\begin{equation}
{\rm tr}_\s\big( \rho_\s(t)  A\big) = \lim_{\lambda\rightarrow\infty} \omega_\s\otimes\omega_\r
\big(e^{i t H} (A\otimes\bbbone_\r) e^{-it H}\big),
\end{equation}
holding for all system observables $A\in\mathcal B(\mathbb C^d)$.

We make the following assumptions on the form factor $g(k)$ and the dispersion $\omega(k)$ in the Hamiltonian \eqref{1},
\begin{itemize}
\item[(A1)] We have the regularity property
\begin{equation}
\label{cond:g}
\mathcal C^{1/2}\frac{e^{i\omega(k)t}-1}{\omega(k)}g(k) \in L^2(\mathbb R^3,d^3k),\quad \forall t>0 \qquad \mbox{and}\qquad \frac{g(k)}{\omega(k)}  \in L^2(\mathbb R^3,d^3k).
\end{equation}

\item[(A2)] We have the effective coupling property
\begin{equation}
\label{gnotzero}
g(k)\neq 0\qquad \mbox{for $k\in\mathbb R^3$ satisfying $a<|k|< b$ for some $0\le a <b$.}
\end{equation}
\end{itemize}

The condition \eqref{cond:g} is used in the so-called  {\it polaron transformation} (see \cite{Marcantoni-Merkli}). However, if one assumes some regularity on the state \eqref{9} then the condition \eqref{cond:g} is not needed. In particular, it is not needed if the reservoir state is the thermal state $\omega_{\r,\beta}$, for any $0<\beta\le \infty$. We explain this further in \cite{Marcantoni-Merkli}. The non-vanishing condition \eqref{gnotzero} on $g$ is satisfied for instance if $g$ is continuous (and not the zero function) | it guarantees that the system is coupled in an effective way to the reservoir.

\begin{thm}[Zeno coupling, \cite{Marcantoni-Merkli}]
\label{thm1}
Assume conditions (A1) and (A2). Then for all $t > 0$, 
\begin{equation}
\label{15}
\rho_\s(t) = e^{-it H_\z} \left( \sum_{l=1}^\nu P_l  \rho_\s  P_l \right) e^{it H_\z},
\end{equation}
where $H_\z$ is the Zeno Hamiltonian (see \eqref{3.0})
\begin{equation}
\label{HZeno}
H_\z = \sum_{l=1}^\nu P_l H_\s P_l.
\end{equation}
\end{thm}

The result shows that in the Zeno coupling limit ($t>0$ fixed, $\lambda\rightarrow\infty$), the system dynamics is that of a non-selective measurement with respect to the measurement observable $G$ (defining the interaction, see \eqref{1}) 
$$
\rho_\s\mapsto  \sum_{l=1}^\nu P_l  \rho_\s  P_l
$$
plus a Schr\"odinger dynamics generated by the Hamiltonian $H_\z$. The dynamics is happening independently on spectral subspaces ${\rm Ran} P_l$ of $G$ | the Hamiltonian $H_\z$ is block diagonal with respect to the decomposition 
$$
\h_\s=\bigoplus_{l=1}^\nu{\rm Ran}P_l.
$$
In the case when $G$ has simple spectrum, $\dim P_l=1$ for all $l$, the blocks of $H_\z$ are $1\times1$ and then \eqref{15} reduces to the time-independent state (Zeno effect)
$$
\rho_\s(t) =  \sum_{l=1}^\nu P_l  \rho_\s  P_l,\quad t>0.
$$
The right side is the state after a non-selective von Neumann projective measurement, associated to the measurement observable $G$, has been performed on $\rho_\s$. The Zeno coupling implements an {\it instantaneous} measurement on the system $\s$.  Unless the initial state $\rho_\s$ is already block diagonal, the measurement causes a discontinuity in the dynamics at $t=0$,
\begin{equation}
\label{n7}
\lim_{t\rightarrow 0_+} \rho_\s(t) \neq \rho_\s.
\end{equation}

Let us try to understand heuristically why the Zeno limit implements instantaneous decoherence. For this we consider a Hamiltonian of the $\s\r$-type \eqref{1}, given by
\begin{equation}
\label{1'}
H(\varepsilon,\lambda) =\varepsilon H_\s +H_\r +\lambda G\otimes X,
\end{equation}
where $\varepsilon$ is another parameter and $X$ is a reservoir operator. Let $\rho_\s(t)$ be the reduced system density matrix at time $t$. It depends on $\varepsilon$ and $\lambda$. Let $\{\psi_k\}$ be an eigenbasis of $G$. The populations of the density matrix $\rho_\s(t)$ in this basis are (by definition) the diagonal density matrix elements, $\langle \psi_k,\rho_\s(t) \psi_k\rangle$. The change of the populations for small times are given by 
\begin{equation}
\label{vareps}
\partial_t|_0 \, {\rm tr}_\s\big(\rho_\s (t) |\psi_k\rangle\langle\psi_k| \big) = -i \omega_{\s\r}\Big(\big[H(\ve,\lambda),|\psi_k\rangle\langle\psi_k|\otimes\bbbone_\r\big] \Big)
 = -i \ve \omega_{\s}\Big(\big[H_\s,|\psi_k\rangle\langle\psi_k|\big] \Big) \propto \ve,
\end{equation}
(we assume $\omega_{\s\r}=\omega_\s\otimes\omega_\r$) while change of the coherences, which (by definition) are the off-diagonal density matrix elements in the $G$-eigenbasis, are given for small times by
\begin{align}
\partial_t|_0 \, {\rm tr}_\s\big(\rho_\s (t) |\psi_k\rangle\langle\psi_l| \big)
&=  -i \omega_{\s\r}\Big(\big[H(\ve,\lambda),|\psi_k\rangle\langle\psi_l|\otimes\bbbone_\r\big] \Big)\nonumber\\
&= -i  \ve \omega_{\s}\Big(\big[ H_\s,|\psi_k\rangle\langle\psi_l|\big] \Big)-i  \lambda (\gamma_k-\gamma_l) \omega_{\s}\big(|\psi_k\rangle\langle\psi_l|\big)  \omega_\r(X)\nonumber\\
&= \ve C_{k,l} +\lambda C'_{k,l}.
\label{varepslambda}
\end{align}
It follows from \eqref{vareps}, \eqref{varepslambda} that for $\lambda>\!\!>\ve$ the time scale $\propto 1/\ve$ at which the populations change is much longer than the time scale $\propto 1/\lambda$ at which the coherences change. (If $C'_{k,l}=0$, say for $\omega_\r(X)=0$ then one looks at the second derivative at zero which is $\propto \lambda^2\omega_\r(X^2)$.) This is the heuristic reason why, in the setting of the Zeno coupling of Section \ref{sec:zenomeasuproc}, where $\ve=1$, $t>0$ fixed and $\lambda\rightarrow\infty$, the coherences of $\rho_\s(t)$ in the $G$ eigenbasis evolve very fast leading to the instantaneous projective measurement and successive Zeno dynamics while the populations are constant (see \eqref{15}).

\smallskip

In order to resolve the dynamics of the measurement process induced by the Zeno limit, we must look at finite but large coupling constants $\lambda$ and short times $t$. This leads us to consider the {\it ultrastrong coupling regime} \eqref{regime0}.

\section{Fine graining of time}
\label{sec:finegraining}

Consider a $\s\r$ bipartite open system with Hilbert space 
\begin{equation}
\label{hsr}
\h_{\s\r} = \h_\s\otimes\h_\r
\end{equation}
with
$$
\dim\h_\s=d<\infty
$$
and Hamiltonian 
\begin{equation}
\label{n1}
H(\lambda)= H_\s+H_\r+\lambda G\otimes X
\end{equation}
where $\lambda\in\mathbb R$, $H_\s$ and $H_\r$ are the Hamiltonians of the system and the reservoir, and where $G$ and $X$ are self-adjoint operators on the Hilbert spaces $\h_\s$ and $\h_\r$ of the system and the reservoir, respectively. Let 
\begin{equation}
\label{n2}
\omega_{\s\r}=\omega_\s\otimes\omega_\r\equiv {\rm tr}_\s(\rho_
\s \,\cdot\,)\otimes {\rm tr}_\r(\rho_\r\,\cdot\,)
\end{equation}
be the initial $\s\r$ state, where the $\rho_\s$  and $\rho_\r$ are density matrices of $\s$ and $\r$.\footnote{For the reservoir typically the state is obtained by a thermodynamic limit and so $\omega_\r$ should generally be viewed as a state in the algebraic sense, that is, a positive normalized functional on an algebra of observables. Then, by the Gelfand-Naimark-Segal construction, $\omega_\r$ can always be represented by a density matrix (of rank one even) on a suitable Hilbert space.} We are interested in the dynamics of the system, described by the density matrix
\begin{equation}
\label{m3}
\rho_\s(t) = {\rm tr}_\r\big( e^{-it H(\lambda)} (\rho_\s\otimes \rho_\r) e^{it H(\lambda)}\big)
\end{equation}
on $\h_\s$. In \eqref{m3} the trace is a partial one, taken over $\h_\r$. In terms of the fine grained time $\tau=\lambda t$, \eqref{regime0}, the propagator of the dynamics \eqref{m3} is 
\begin{equation}
\label{n9}
e^{-i t H(\lambda)} = e^{-i\tau \big(\frac{ H_\s+H_\r}{\lambda} +G\otimes X\big)}.
\end{equation}
We obtain formally,
\begin{equation}
\label{n10'}
\llim e^{-it H(\lambda)} =  e^{-i\tau G\otimes X}.
\end{equation}
For bounded operators $H_\r, X$ ($H_\s,G$ are automatically bounded as we assume $\dim\h_\s<\infty$) the equality \eqref{n10'} in the sense of operator norm convergence follows directly from ($H_0=H_\s+H_\r$)
$$
e^{-i\tau(\frac{H_0}{\lambda}+ G\otimes X)} -  e^{-i\tau G\otimes X} =  \frac{-i}{\lambda}  \int_0^\tau e^{-i s (\frac{H_0}{\lambda} +G\otimes X)} H_0 e^{i(s-\tau) G\otimes X} ds.
$$
However, we are usually interested in unbounded reservoir operators $H_\r$ and $X$, and the precise meaning of the limit \eqref{n10'} needs to be elucidated. We do this later in Section \ref{sec:spinboson} for spin-Boson type models. For now we examine some implications of \eqref{n10'}. From \eqref{m3} and \eqref{n10'} we obtain for $\tau>0$,
\begin{equation}
\label{n11}
\varrho_\s(\tau)\equiv \llim \rho_\s(t) ={\rm tr}_\r\big(e^{-i\tau G\otimes X} (\rho_\s\otimes \rho_\r) e^{i\tau G\otimes X}\big).
\end{equation}
Using the diagonal form \eqref{3.0} of $G$ gives
\begin{equation}
\label{n12}
\varrho_\s(\tau) =\sum_{l,r =1}^\nu  D_{l,r}(\tau)P_l \rho_\s P_r,\qquad \mbox{with}\qquad  D_{l,r}(\tau)\equiv {\rm tr}_\r \big(\rho_\r\, e^{-i\tau (\gamma_l-\gamma_r) X}\big).
\end{equation}
As $D_{l,r}(0)=1$ the continuity at time zero is restored,
\begin{equation}
\label{n13}
\lim_{\tau\rightarrow 0_+}  \varrho_\s(\tau)=\rho_\s. 
\end{equation}
Assume that $\lim_{|s|\rightarrow\infty} {\rm tr}_\r(\rho_\r\, e^{isX})=0$. Then \eqref{n12} shows that
\begin{equation}
\lim_{\tau\rightarrow\infty}\varrho_\s(\tau) = \sum_l P_l\rho_\s  P_l.
\end{equation}
In the $\tau\rightarrow\infty$ limit we thus recover the Zeno result (albeit with $H_\s=0$, as the effect of $H_\s$ is eliminated by taking $t\rightarrow 0$ in the current scaling, \eqref{n11}). The function  
\begin{equation}
\label{31'}
D(s) = \omega_\r(e^{-isX}),\quad s\in\mathbb R, 
\end{equation}
is of positive type, meaning that the matrix $D$ with elements $D_{ij}=D(s_i-s_j)$ is a positive semidefinite matrix for any $s_1,\ldots,s_N$ and any $N\ge 1$.\footnote{Let $N\ge 1$, denote the inner product of $\mathbb C^N$ by $\langle\cdot,\cdot\rangle$ and let $z=(z_1,\ldots,z_N)^T\in\mathbb C^N$. Then  $\langle z,  D z\rangle= \omega_\r( Y^\dag Y)\ge 0$, where $Y=\sum_{l=1}^N z_l e^{is_l X}$. So $D$ is positive semidefinite.} Moreover, $D(0)=1$. Suppose that $D$ is continuous. Then it follows from  Bochner's theorem that there exists a unique Borel probability measure $\mu$ on $\mathbb R$ such that
\begin{equation}
\label{bb31}
D(s) = \int_{\mathbb R} d\mu(\alpha) e^{-i\alpha s}.
\end{equation}
If $D(s)$ is a square integrable function, that is,
$$
\int_{\mathbb R} |D(s)|^2ds <\infty,
$$
then we have 
\begin{equation}
\label{mm39}
d\mu(\alpha) = \mu(\alpha) d\alpha,\qquad \mu(\alpha) = \frac{1}{2\pi}\int_{\mathbb R} e^{i\alpha s} D(s) ds. 
\end{equation}
This follows from \eqref{bb31} and the Fourier inversion theorem on $L^2(\mathbb R,ds)$ and the uniqueness of $d\mu(\alpha)$ for a given $D(s)$. We summarize these findings in the following result. 
\begin{thm}
\label{thm:2}
For all $\tau>0$, the state $\varrho_\s(\tau)$, \eqref{n11} has the expression 
\begin{align}
\varrho_\s(\tau) & = \sum_{l,r=1}^\nu D_{l,r}(\tau) P_l \rho_\s P_r\label{pop}\\
& = \int_{\mathbb R}  e^{i\tau\alpha G}\rho_\s \, e^{-i\tau\alpha G}\, d\mu(\alpha) \ \ \qquad \mbox{(if $D(s)$ is continuous)}\label{pop1} \\
&= \int_\rx \mu(\alpha)  e^{i\tau\alpha G}\rho_\s\, e^{-i\tau\alpha G}\, d\alpha \qquad \mbox{(if $D(s)$ is continuous and square integrable)}
\label{pop2}
\end{align}
where $D_{l,r}(\tau)= D(\tau(\gamma_l-
\gamma_r))$, $\mu$ the measure given by \eqref{bb31}, which satisfies \eqref{mm39} for square integrable $D$.
\end{thm}

According to \eqref{pop} the populations (in the eigenbasis of $G$) are time-independent; this can be understood as a consequence of the fact that in the considered limit $t\rightarrow 0$, the system Hamiltonian $H_\s$ does not play any role (see also \eqref{HH}). The second and third equality in \eqref{pop} show that $\varrho_\s(\tau)$ is a convex combination of density matrices $e^{i\tau\alpha G}\rho_\s e^{-i\tau\alpha G}$, each following a Markovian dynamics at a different speed given by the scaling of time, $\tau\alpha$ for varying $\alpha$. It is known that a convex combination of Markovian dynamics is generally non-Markovian \cite{WECC}. We now investigate the Markovianity properties of $\Lambda(\tau): \rho_\s \to \varrho_\s(\tau)$ more closely.

\section{Markovianity in the fine grained time regime}
\label{sec:markov}

\subsection{General Results}

Motivated by \eqref{pop} we consider the dynamics of density matrices $\rho$ of dimension $d \times d$ given by a dynamical map $\Lambda(\tau)$ of the form 
\begin{equation}\label{eq:maplambda}
\Lambda(\tau)\rho = \sum_{l,r=1}^\nu D_{l,r}(\tau) P_l \rho P_r,
\end{equation}
where $\tau\ge 0$ and the $P_l$ are a complete family of orthogonal projections (possibly of rank higher than one, in which case $\nu < d$), 
\begin{equation}
\label{pfamily}
P^\dag_l=P_l,\quad P_l P_r=\delta_{l,r} P_l, \quad \sum_{l=1}^\nu P_l =\bbbone.
\end{equation}
We assume that for each $\tau\ge 0$, $D_{l,r}(\tau)$ defines a positive definite matrix, that  $D_{l,r}(0)=1$ and that $D_{l,l}(\tau)=1$.
As we show below, these assumptions imply that  $\Lambda(\tau)$ is completely positive, that $\Lambda(0)=\bbbone$ and that $\Lambda(\tau)$ is trace preserving, respectively. 
In order to study the Markovianity of $\Lambda(\tau)$, we further assume that $D_{l,r}(\tau)\neq 0$ for all $l,r$ and $\tau\ge 0$, so that the inverse map $\Lambda(\tau)^{-1}$ exists\footnote{The characterization of Markovianity for non-invertible maps is more involved, see for instance \cite{noninvertible}.} for all $\tau\ge 0$ and reads explicitly
\begin{equation}\label{eq:inverse}
\Lambda(\tau)^{-1}\rho = \sum_{l,r =1}^\nu \big(D_{l,r}(\tau)\big)^{-1} P_l \rho P_r \,.
\end{equation}
An invertible map $\Lambda(\tau)$ is called CP-divisible if the intermediate propagator 
\begin{equation}
\label{n16}
V(\tau,\sigma) :=\Lambda(\tau)\Lambda(\sigma)^{-1} 
\end{equation}
is completely positive for all $\tau\ge \sigma \ge 0$. CP-divisibility is one of the definitions of quantum Markovianity \cite{divisibility}. If $V(\tau,\sigma)$ is positivity preserving then $\Lambda(\tau)$ is called P-divisible. Maps that are P-divisible but not CP-divisible are called weakly non-Markovian. If there is a $\tau\ge\sigma\ge 0$ such that $V(\tau,\sigma)$ is not positive, then the dynamics is called essentially non-Markovian.

Based on \eqref{eq:maplambda} and \eqref{eq:inverse} we have 
\begin{equation}
\label{n17}
   V(\tau,\sigma) \rho  = \sum_{l,r =1}^\nu Q_{l,r} (\tau,\sigma) \, P_l \rho P_r, \quad \text{where} \quad Q_{l,r} (\tau,\sigma) := \frac{D_{l,r}(\tau)}{D_{l,r}(\sigma)} ,
\end{equation}
so that the map $V(\tau,\sigma)$ has the same structure as $\Lambda(\tau)$ with a different coefficient matrix. For this class of dynamics and rank-one projections $P_l$, the Markovianity properties have been studied in \cite{Lonigro22} (Proposition 4.1). Our next result is a slight generalization to arbitrary rank.

\begin{prop}[Slight generalization of Lonigro-Chru\'{s}ci\'{n}ski \cite{Lonigro22}] \label{thm_markov}
Let $P_l$, $l=1,\ldots,\nu$ be a family of projections on $\cx^d$ satisfying \eqref{pfamily} and let $A=(A_{l,r})$ be a $\nu\times\nu$ complex matrix. Consider the map  $\Pi_A\rho = \sum_{l,r=1}^\nu A_{l,r}P_l \rho P_r$, where $\rho$ is an operator on $\cx^d$. The following are equivalent:
  \begin{itemize}
        \item[(a)] $\Pi_A$ is {\rm CP}
        \item[(b)] $\Pi_A$ is {\rm P}
        \item[(c)] $A$ is positive semi-definite 
    \end{itemize} 
\end{prop}

\noindent
{\it Proof of Proposition \ref{thm_markov}.} The proof follows closely \cite{Lonigro22}. To see that (c) $\Rightarrow$ (a) we note that for a positive definite $A$, the formula $\sum_{l,r=1}^\nu A_{l,r}P_l (\cdot) P_r$ is a (non-diagonal) Kraus representation for $\Pi_A$ and therefore $\Pi_A$ is CP. By definition, $(a) \Rightarrow (b)$. Finally, the implication $(b) \Rightarrow (c)$ is proved by showing the contrapositive. This is the only part where the rank of the projections plays a role.
We will show that if $A$ is not positive semi-definite, then there are vectors  $x=(x_1,\ldots,x_d)$, $y=(y_1,\ldots,y_d)\in\mathbb C^d$ such that $\langle x, \big(\Pi_A|y\rangle\langle y|\big) x\rangle <0$. We consider the canonical basis as the one in which the $P_l$ are diagonal. We have 
$$
\langle x, \big(\Pi_A|y\rangle\langle y|\big) x\rangle =\sum_{l,r=1}^\nu A_{l,r}\langle x,P_l y\rangle \langle y, P_rx\rangle = \langle v, A v\rangle,
$$
where $v=(v_1,\ldots,v_{\nu})\in \mathbb C^{\nu}$ is the vector with components $v_k=\langle y,P_k x\rangle$. (Note that the last inner product above is that of $\cx^\nu$ while other ones are that of $\cx^d$.) Suppose that $A$ is not positive definite. Then there is a (normalized) $v\in\mathbb C^{\nu}$ such that $\langle v, A v\rangle<0$. Choosing $x=(1,\ldots,1)$ we have $\langle y,P_kx\rangle = \sum_{j\in J_k} y_j$, where $J_k$ is the set of indices associated with the coordinates of the subspace ${\rm Ran}P_k$. Now we take $y_j=\frac{v_k}{|J_k|}$ for $j\in J_k$, and we do this for all $k$. Then  $\langle y, P_kx\rangle=v_k$ for all $k$. It follows that $\langle x, (\Pi_A|y\rangle\langle y|) x\rangle<0$, which shows that $\Pi_A|y\rangle\langle y|$ is not positive definite, hence $\Pi_A$ is not positivity preserving. This completes the proof of Proposition \ref{thm_markov}. \hfill $\qed$
\medskip

Applying the Proposition \ref{thm_markov} to $A=(Q_{l,r}(\tau,\sigma))$ given in \eqref{n17} we obtain the next result.

\begin{cor}
\label{cor_markov}
The following statements about the dynamics $\Lambda(\tau)$, \eqref{eq:maplambda} are equivalent:
  \begin{itemize}
        \item[(a)] $\Lambda(\tau)$ is CP-divisible
        \item[(b)] $\Lambda(\tau)$ is P-divisible
        \item[(c)] The $\nu\times \nu$ matrix $Q = \big( Q_{l,r} (\tau,\sigma) \big)_{l,r}$ defined in \eqref{n17} is positive semi-definite for every $\tau\ge\sigma\ge0$
    \end{itemize} 
A necessary condition for CP-divisibility is $|D_{l,r}(\tau)| \leq |D_{l,r}(\sigma)|$ for $\tau \geq \sigma \geq 0$. This is also sufficient for $d=2$.
\end{cor}

{\it Proof of Corollary \ref{cor_markov}.} The equivalences follow directly from Proposition \ref{thm_markov} applied to the map $\Pi_{Q(\tau,\sigma)}=V(\tau,\sigma)$. Moreover, the positive semi-definiteness of $Q$ implies that all principal minors are nonnegative (see e.g.~the equivalent characterizations of positive semi-definite matrices in \cite{book_matrix} p.~566). In particular, all $2 \times 2$ principal minors are nonnegative, i.e. for any pair of indices $i_1,i_2$ we have
\begin{equation}\label{eq:2Ddeterminant}
    Q_{i_1,i_1}(\tau,\sigma) Q_{i_2,i_2}(\tau,\sigma) - Q_{i_1,i_2}(\tau,\sigma) Q_{i_2,i_1}(\tau,\sigma) \geq 0.
\end{equation}
Using \eqref{n17}, the hermiticity $D_{l,r}(\tau)=\overline{D_{r,l}(
\tau)}$ and $D_{r,r}(\tau)=1$, we obtain from \eqref{eq:2Ddeterminant} that $|Q_{l,r}(\tau,\sigma)|^2\le 1$, which implies $|D_{l,r}(\tau)| \leq |D_{l,r}(\sigma)|$ for $\tau \geq \sigma \geq 0$. Finally, for $d=2$, $Q$ is a $2 \times 2$ matrix with trace equal to $2$. Therefore, it is positive semi-definite whenever the determinant is nonngegative, which in turn is guaranteed by the monotonicity of $|D_{l,r}(\tau)|$, as we just discussed. Corollary \ref{cor_markov} is proven. \hfill $\qed$
\bigskip

{\bf Counterexample:~Monotonicity without CP-divisibility.} We now show that already for $d=3$ the monotonicity condition in Corollary \ref{cor_markov} is not sufficient for CP-divisibility. 
The matrix 
\begin{equation}
\label{D}
    D = 
    \begin{pmatrix}
      1 & a & b \\
      a & 1 & a \\
      b & a & 1
    \end{pmatrix},\qquad a,b\ge 0,
\end{equation}
has determinant $\det D=(1-b)(1+b-2a^2)$. Let $D(\tau)$ be the matrix $D$ with $a(\tau)=e^{-\tau/4}$ and $b(\tau)=(1+\tau^3/2)^{-1}$. One easily sees that $D(\tau)\ge 0$ for all $\tau\ge 0$ (check for instance that all the principal minors are non-negative) and furthermore, the matrix elements $D_{l,r}(\tau)$ satisfy $|D_{l,r}(\tau)| \leq |D_{l,r}(\sigma)|$ for $\tau \geq \sigma \geq 0$. However, there are $\tau\ge\sigma\ge0$ such that the matrix $Q(\tau,\sigma)$ with the elements $Q_{l,r} (\tau,\sigma) = D_{l,r}(\tau)/D_{l,r}(\sigma)$ is not positive semi-definite. Indeed, $Q(\tau,\sigma)$ is again of the form \eqref{D} above, with matrix elements $a(\tau,\sigma)=e^{-(\tau-\sigma)/4}$ and $b(\tau,\sigma)=(2+\sigma^3)/(2+\tau^3)$. The $2\times2$ principal minors of $Q(\tau,\sigma)$ are always non-negative, therefore $Q(\tau,\sigma)$ is positive semi-definite if and only if $\det Q(\tau,\sigma) \geq 0$. Now 
\begin{equation*}
    \det Q(\tau,\sigma) = \left( 1-\frac{2+\sigma^3}{2+\tau^3} \right)\left( 1+ \frac{2+\sigma^3}{2+\tau^3} - 2 e^{-(\tau-\sigma)/2}\right)
\end{equation*}
and the first factor is strictly positive for all $\tau> \sigma$. For $\tau=2$ and $\sigma=3/2$ the second factor is $<-1/50$. It follows that $Q(2,3/2)$ is not positive semi-definite. 
\medskip

{\bf Remark.} The counterexample we provided has a different qualitative time-dependence for the different matrix elements of $D(\tau)$. However, dynamical maps that are obtained in the fine-grained-time regime from the interaction with a reservoir correspond to $D_{l,r}(\tau)= \omega_\r(e^{-i \tau (\gamma_l-\gamma_r) X})$, so that the different matrix elements share the same functional form of the time-dependence. Finding a counterexample within this class seems more difficult. In particular, for a natural ansatz like $D_{l,r}(\tau)= e^{- a_{lr} \tau^\alpha}$, with non-negative coefficients $a_{lr}$ and some exponent $\alpha >0$, one has that $D(\tau)\geq 0 \Rightarrow Q(\tau,\sigma) \geq 0$ for $\tau \geq \sigma \geq 0$. Indeed, $Q(\tau,\sigma)= D(\sqrt[\alpha]{\tau^\alpha -\sigma^\alpha})$.

\bigskip

\subsection{Markovianity properties via the generator}

An alternative way of studying Markovianity is to investigate the master equation, that is the dynamical equation for the time-dependent density matrix $\Lambda(\tau)\rho$, 
\begin{equation}
\label{eq:mastereq}
    \partial_\tau \Lambda(\tau)\rho = \mathcal L_\tau \Lambda(\tau)\rho ,
\end{equation}
where the time-dependent generator is defined by 
\begin{equation}\label{generator}
    \mathcal{L}_\tau (\cdot) := \big(\partial_\tau \Lambda(\tau) \big) \Lambda(\tau)^{-1} (\cdot),
\end{equation}
and we of course assume that $\Lambda(\tau)$ is differentiable in $\tau$ and invertible as before. By construction, as $\Lambda(\tau)$ is trace-preserving and hermiticity-preserving (that is, ${\rm tr}(\Lambda(\tau)A)= {\rm tr}(A)$ and $\Lambda(\tau)(A^\dag) =\big( \Lambda(\tau) (A)\big)^\dag$), the generator $\mathcal{L}_\tau$ satisfies 
\begin{equation}\label{hptp}
   (i) \  {\rm tr}\big(\mathcal{L}_\tau(A) \big) =0, \qquad \text{and} \qquad (ii) \  \mathcal{L}_\tau (A^\dag) =\big( \mathcal{L}_\tau (A)\big)^\dag.
\end{equation}
When talking about generators in the following, we always consider linear maps having the properties (i) and (ii) in \eqref{hptp}. There are many equivalent ways of expressing a given generator. Some of those are called canonical \cite{review}.
\begin{definition} 
A {\rm canonical form} of a (possibly time-dependent) generator $\mathcal{L}_\tau: M_d(\mathbb{C}) \to M_d(\mathbb{C})$ is 
\begin{equation}\label{canonical}
    \mathcal{L}_\tau (\cdot)= -i[H(\tau), \cdot] + \sum_{l,r=1}^{d^2-1} \widehat{C}_{lr}(\tau) \Big( F_l(\tau) (\cdot) F_r^\dag(\tau) - \frac{1}{2} \big\{F_r^\dag(\tau) F_l(\tau), \cdot \big\} \Big) ,
\end{equation}
where for each $\tau\ge 0$, $H(\tau)$ is a hermitian $d\times d$ matrix, $\widehat{C}(\tau):=(\widehat{C}_{lr}(\tau))_{lr}$ is a hermitian $(d^2-1) \times (d^2-1)$ matrix, and the set $\{ F_i (\tau) \}_{i=1}^{d^2-1} \cup \{ \bbbone /\sqrt{d} \}$ is a orthonormal basis of $M_d(\mathbb{C})$. Equivalently, the operators $F_i(\tau)$ are orthonormal and traceless. 
\end{definition}
The orthonormality is with respect to the Hilbert-Schmidt inner product, that is, here and in the following we consider the bounded operators $M_d(\cx)$ on $\cx^d$ with inner product $\langle X,Y\rangle={\rm tr}(X^\dag Y)$ and induced norm $\|X\|^2=\tr(X^\dag X)$.
It is readily checked that the form \eqref{hptp} implies (i) and (ii) of \eqref{hptp} | in particular, the hermiticity of $\widehat{C}(\tau)$ implies the hermiticity-preservation (ii). The canonical form of a given generator is not unique. All canonical forms are related by a change of basis that keeps one element fixed, the multiple of the identity $\bbbone /\sqrt{d}$. Therefore, the coefficient matrix $\widehat{C}(\tau)$ |  called Kossakowski matrix |  depends on the specific canonical form considered. However, some properties, such as hermiticity or semi-definiteness of the Kossakowski matrix, are representation independent.

\begin{definition}[\cite{review}]
A (possibly time-dependent) generator $\mathcal{L}_\tau: M_d(\mathbb{C}) \to M_d(\mathbb{C})$ is in {\rm diagonal canonical} form if it is written as \eqref{canonical} with a diagonal matrix $\widehat{C}(\tau)$. Equivalently, $\forall \tau\ge0$
\begin{equation}\label{canonicaldiag}
    \mathcal{L}_\tau (\cdot)= -i[H(\tau), \cdot] + \sum_{i=1}^{d^2-1} \Gamma_i(\tau) \Big( \widehat{F}_i(\tau) (\cdot) \widehat{F}_i^\dag(\tau) - \frac{1}{2} \big\{\widehat{F}_i^\dag(\tau) \widehat{F}_i(\tau), \cdot \big\} \Big) 
\end{equation}
for real numbers $\{ \Gamma_i(\tau) \}_{i=1}^{d^2-1}$ and orthonormal traceless operators $\{ \widehat{F}_i (\tau) \}_{i=1}^{d^2-1}$.
\end{definition}

A generator given in canonical form can always be put in diagonal canonical form:  As $\widehat{C}(\tau)$ is hermitian we have for any fixed $\tau$,  $\widehat{C}(\tau)= U^\dag(\tau) \Gamma(\tau) U(\tau)$ for some unitary matrix $U(\tau)$ and diagonal matrix $\Gamma(\tau)$ with real entries $\Gamma_{ij}(\tau)=\Gamma_i (\tau)\delta_{ij}$. Defining $\widehat{F}_i(\tau):= \sum_{l} U_{li} F_l(\tau)$ one passes from \eqref{canonical} to \eqref{canonicaldiag}. Coming back to our original motivation of studying Markovianity, a well-known result connecting the structure of the generator with Markovianity is the following.
\begin{thm}[Corollary 7 in \cite{review}]\label{thm_generator}
   A time-dependent generator $\mathcal{L}_\tau$ written in diagonal canonical form \eqref{canonicaldiag} defines a CP-divisible dynamics if and only if all the coefficients $\Gamma_i(\tau)$ are non-negative for all times $\tau$.
\end{thm}
According to Theorem \ref{thm_generator}, a strategy to characterize Markovianity is to write the generator in a canonical form \eqref{canonical}, diagonalize the matrix $\widehat{C}(\tau)$ in order to arrive at a diagonal canonical form and check the sign of the eigenvalues $\Gamma_i(\tau)$. In particular, the positive-definiteness of the matrix $\widehat{C}(\tau)$ will not depend on the specific canonical form considered. 

Let us now focus the discussion on dynamics of the form \eqref{eq:maplambda}, like those emerging under strong coupling in the fine grained time regime. Explicitly, from the definition \eqref{generator} one can compute
\begin{equation}\label{P_generator}
    \mathcal{L}_\tau (\cdot) =\sum_{l,r=1}^\nu C_{lr}(\tau) P_l (\cdot) P_r ,\qquad \text{where} \qquad C_{lr}(\tau) = \frac{\partial_\tau D_{l,r}(\tau)}{D_{l,r}(\tau)}.
\end{equation} 
Since $\partial_\tau D_{l,l}(\tau)=0$ for all $l,\tau$, the sum in \eqref{generator} is effectively only over $l\neq r$. Also, for $l \neq r$ one has $P_l P_r=0$ because these are orthogonal projections. Therefore, \eqref{P_generator} can be written equivalently
\begin{equation}\label{P2_generator}
    \mathcal{L}_\tau (\cdot)= \sum_{l,r=1}^\nu C_{lr}(\tau) \Big( P_l (\cdot) P_r - \frac{1}{2} \big\{P_r P_l, \cdot \big\} \Big) ,
\end{equation}
that at a first look resembles a canonical form. However, the operators $P_r$, despite being orthogonal, are not normalized, $\langle P_r,P_r\rangle:={\rm tr} P^\dag_r P_r\neq 1$ (unless the $P_r$ are rank-one) and moreover they are not traceless. Therefore, this is not a canonical form.

To write the generator in a canonical form, we use a Gram-Schmidt orthogonalization procedure. The set of $\nu$ operators $\{\bbbone, P_1,\ldots,P_{\nu-1}\}$ is linearly independent because the $P_j$ are spectral projections \eqref{3.0} satisfying $\sum_{j=1}^\nu P_j=\bbbone$. By Gram-Schmidt we construct the orthonormal set  $\{ F_i\}_{i=0}^{\nu-1}$ as
\begin{equation}
F_0 = \frac{\bbbone}{\sqrt{d}}, \qquad F_i=\frac{\widetilde F_i}{\|\widetilde F_i\|} \qquad \text{with}\qquad \widetilde{F}_i = P_i - \sum_{j=0}^{i-1}   \frac{{\rm tr}(P_i \widetilde{F}_j)}{\| \widetilde{F}_j \|^2}\widetilde{F}_j, \quad  1\leq i \leq \nu-1. 
\end{equation}
This gives
\begin{equation}
    P_i = \|\widetilde F_i\| F_i + \sum_{j=0}^{i-1}  F_j\, {\rm tr}(P_i F_j), \quad i=1,\ldots,\nu-1,\qquad P_\nu=\bbbone-\sum_{j=1}^{\nu-1}P_j=\sqrt d F_0-\sum_{j=1}^{\nu-1}P_j.
\end{equation}
Substituting these expressions for $P_l, P_r$ into \eqref{P2_generator} we obtain
\begin{align}\label{finalform}
    \mathcal{L}_\tau (\cdot) &= \sum_{l,r = 0}^{\nu-1} \widehat{C}^0_{lr} (\tau) \Big( F_l (\cdot) F_r - \frac{1}{2} \big\{F_r F_l, \cdot \big\} \Big) \nonumber \\
    &= \sum_{l,r = 1}^{\nu-1} \widehat{C}^0_{lr}(\tau)  \Big( F_l (\cdot) F_r - \frac{1}{2} \big\{F_r F_l, \cdot \big\} \Big) + i \sum_{r=1}^{\nu-1} \frac{{\rm Im} \, \widehat{C}^0_{r0}(\tau)}{\sqrt{d}} \big[ F_r, \cdot \,\big]
\end{align}
where the explicit expression for $\widehat{C}^0_{lr}(
\tau)$ can be obtained in principle from ${C}_{lr}(\tau)$ by exploiting the recursive definition of the $F_i$. The last equality follows from isolating the terms with $l=0$, $r=0$: the term with $l=r=0$ vanishes; for  $l=0$ and $r\neq 0$ one has for any matrix $A$
\begin{equation*}
    \sum_{r=1}^{\nu-1} \frac{\widehat{C}^0_{0r}(\tau)}{\sqrt{d}} \Big(  A F_r - \frac{F_r A}{2} - \frac{A F_r}{2} \Big) = - \sum_{r=1}^{\nu-1} \frac{\widehat{C}^0_{0r}(\tau)}{\sqrt{d}} \frac{\big[ F_r , A\big] }{2}. 
\end{equation*}
The terms with $l\neq 0$ and $r=0$ give the same expression as on the right side, with $\widehat C^0_{0r}(\tau)$ replaced by $-\widehat C^0_{r0}(\tau)$. Exploiting the hermiticity\footnote{The hermiticity of $\widehat{C}^0(\tau)$ is guaranteed by the property \eqref{hptp}-(ii), see e.g. \cite{review}.} of $\widehat{C}^0(\tau)$ one arrives at \eqref{finalform}.

Note that the $(\nu-1)\times (\nu-1)$ matrix $\widehat{C}(\tau)$, that is obtained from $\widehat{C}^0(\tau)$ eliminating the elements with $l=0$ or $r=0$, is also hermitian.
The generator \eqref{finalform} is in canonical form\footnote{The orthonormal set $\{ F_i\}_{i=0}^{\nu-1}$ is not a basis because the dimension of $M_{d}(\mathbb{C})$ is $d^2>\nu$. However, it becomes a basis once completed with arbitrary orthonormal elements $\{ F_i\}_{i=\nu}^{d^2-1}$ that do not play a role in the considered dynamics. Therefore, the $(d^2-1)\times (d^2-1)$ Kossakowski matrix has a block-diagonal structure 
$\begin{pmatrix}
  [\widehat{C}]_{(\nu-1) \times (\nu-1)} & [0]_{(\nu-1) \times (d^2-\nu) } \\
  [0]_{ (d^2-\nu) \times (\nu-1)  } & [0]_{ (d^2-\nu) \times (d^2-\nu) }
\end{pmatrix} $ and it is positive semi-definite if and only if $\widehat{C}$ is. } and therefore, according to Theorem \ref{thm_generator}, the positivity of $\widehat{C}(\tau)$ for any $\tau$ gives a further explicit criterion to determine the Markovianity of the dynamics \eqref{eq:maplambda}. 
\medskip

{\bf Remark.} We can normalize the projections  in \eqref{P_generator}, \eqref{P2_generator} and introduce the orthonormal family $\widetilde{P}_l := P_l/ \sqrt{{\rm tr}(P_l)}$  to arrive at 
\begin{equation}
\label{naiveform}
    \mathcal{L}_\tau (\cdot) 
    = \sum_{l,r=1}^\nu \widetilde{C}_{lr}(\tau) \Big( \widetilde{P}_l (\cdot) \widetilde{P}_r - \frac{1}{2} \big\{ \widetilde{P}_r \widetilde{P}_l, \cdot \big\}  \Big) , 
\end{equation}
with $\widetilde{C}_{lr}(\tau):= \sqrt{{\rm tr}(P_l){\rm tr}(P_r)}\, C_{lr}(\tau)$. Still, this is not a canonical form of the generator, because the $\widetilde P_j$ are not traceless. The matrix $\widetilde{C}(\tau)$ is hermitian and traceless (all diagonal elements are vanishing) and  so unless $\widetilde{C}(\tau)$ is the zero matrix, it cannot be positive semi-definite. A naive (and wrong!) application of Theorem \ref{thm_generator} would lead  us to conclude that the dynamics is not CP-divisible. This is not true however, as we show in the following example.
\medskip

{\bf Example.}  Consider the case of a free Bosonic reservoir as in Section \ref{sec:zenomeasuproc} in a centered Gaussian state $\omega_\r$, with interaction operator (see \eqref{n1}) $X= \varphi(g)$. The decoherence function is
\begin{equation*}
    D_{l,r}(\tau) := \omega_\r \Big( W \big(\tau (\gamma_r - \gamma_l)g \big) \Big) = e^{-\frac{\tau^2}{4}(\gamma_l -\gamma_r)^2 \langle g, C g \rangle}
\end{equation*}
and the elements of the matrix $Q(\tau,\sigma)$ are 
\begin{equation}
    Q_{l,r} (\tau,\sigma) = e^{-\frac{\tau^2-\sigma^2}{4}(\gamma_l -\gamma_r)^2 \langle g, C g \rangle} = D_{l,r}(\sqrt{\tau^2-\sigma^2}). 
\end{equation}
$Q(\tau,\sigma)$ positive semi-definite for any $\tau\ge \sigma\ge0$ because $D(s)$ is positive semi-definite for all $s\ge 0$. Corollary \ref{thm_markov} implies that the dynamics $\Lambda(\tau)$ is Markovian. However, it is not of semigroup type, because $\Lambda(\tau-\sigma)\Lambda(\sigma) \neq \Lambda(\tau)$, as can be easily checked from the definition \eqref{eq:maplambda}. {\it We have shown that the fine-grained time scaling can lead to a dynamics that does not satisfy the semigroup composition law but is still Markovian in the sense of CP-divisibility.} This has to be compared with other scaling regimes studied in the literature, such as the weak coupling, singular coupling, or low density scaling, which all lead to a semigroup dynamics.
\medskip

We further develop the above example to illustrate the connection between Markovianity and the form of the generator. We already know that the dynamics is Markovian, so the Kossakowski matrix in any canonical form of the generator will be positive-semidefinite. Let us see how to get to a canonical form without using the Gram-Schmidt procedure in this specific case. Since all the elements $D_{l,r}(\tau)$ are strictly positive, we can rewrite the generator \eqref{generator} as
\begin{equation}
\label{56}
    \mathcal{L}_\tau (\cdot)= \sum_{ \substack{1 \leq l, r \, \leq \nu \\ l\neq r}  } \partial_\tau [\ln(D_{l,r}(\tau))] \ P_l (\cdot) P_r = -\tfrac\tau 2 \|\sqrt\mathcal C g\|^2 \sum_{\substack{1 \leq l, r \, \leq \nu \\ l\neq r} } (\gamma_r - \gamma_l)^2 \ P_l (\cdot) P_r.
\end{equation}
Let us assume that the projections $P_l=|l\rangle\langle l|$ are rank 1, so that $1 \leq l \leq \nu= d = \dim(\h_\s)$, and make the ansatz
\begin{equation}\label{ansatz}
    \mathcal{L}_\tau (\cdot) = \sum_{l=1}^d \eta_{ll}(\tau) \left( P_l (\cdot) P_l - \frac{1}{2} \{ P_l , \cdot \} \right) +  \sum_{\substack{1 \leq l, r \, \leq d \\ l\neq r} } \eta_{lr}(\tau) P_l (\cdot) P_r,
\end{equation}
with real coefficients $\eta_{lr}$ that satisfy $\eta_{lr}=\eta_{rl}$.
We obtain on the one hand from \eqref{ansatz},
\begin{equation}\label{comp1}
  \mathcal{L}_\tau (\ket{j}\bra{k})  = \left(-\tfrac{1}{2}(\eta_{jj} + \eta_{kk}) + \eta_{jk} \right) \ket{j}\bra{k}= -\tfrac{1}{2}\left( \eta_{jj} + \eta_{kk} -2 \eta_{jk} \right)  \ket{j}\bra{k}
\end{equation}
and on the other hand from \eqref{56},
\begin{equation}\label{comp2}
    \mathcal{L}_\tau (\ket{j}\bra{k}) = -\frac\tau 2 \|\sqrt\mathcal C g\|^2 (\gamma_k - \gamma_j)^2 \ket{j}\bra{k}.
\end{equation}
Comparing \eqref{comp1} and \eqref{comp2} gives
\begin{equation}\label{consistency}
    \tau \|\sqrt\mathcal C g\|^2 (\gamma_k - \gamma_j)^2 = \eta_{jj} + \eta_{kk} -2 \eta_{jk} , \qquad k,l =1,\ldots d.
\end{equation}
For $j=k$ the expression \eqref{consistency} is satisfied without giving constraints on $\eta_{jk}$, and the symmetry of the left side implies that $\eta_{kj}=\eta_{jk}$. We have thus $(d^2+d)/2$ unknowns $\eta_{jk}$ with $1\le j\le k\le d$ for the $(d^2-d)/2$ independent equations \eqref{consistency} with $1\le j < k \le d$. This means that we have an underdetermined system and so there is some freedom in choosing the $\eta_{jk}$. 
With the choice $\eta_{jk}(\tau) = \tau \|\sqrt\mathcal C g\|^2 \gamma_j \gamma_k$, $j,k=1,\ldots,d$ the expression \eqref{ansatz} becomes
\begin{align*}
    \mathcal{L}_\tau (\cdot) = \tau \|\sqrt\mathcal C g\|^2  \sum_{l,r=1}^d \gamma_l \gamma_r  \left( P_l (\cdot) P_r - \frac{1}{2} \{ P_l P_r , \cdot \} \right) = \tau \|\sqrt\mathcal C g\|^2 \left( L' (\cdot) L' - \frac{1}{2} \{ (L')^2, \cdot\} \right),
\end{align*}
with a single global jump operator $L'= \sum_{l=1}^d \gamma_l P_l$. We shift it to the traceless $L''=L'- \sum_{l=1}^d\gamma_l $ and then normalize it $L= L''/\sqrt{\sum_l\gamma_l^2 +(d-2)(\sum_l\gamma_l)^2}$, so that we arrive at the generator 
\begin{equation}
    \mathcal{L}_\tau (\cdot) = \Gamma(\tau) 
       \Big( L (\cdot) L - \frac{1}{2} \{ L^2, \cdot\} \Big),\qquad \Gamma(\tau) = \tau \|\sqrt\mathcal C g\|^2 \left(\sum_{l=1}^d \gamma_l^2 + (d-2)\Big(\sum_{l=1}^d\gamma_l\Big)^2 \right).
\end{equation}
This is a canonical diagonal form of $\mathcal L_\tau$. Theorem \ref{thm_generator} implies that the dynamics is CP-divisible.

\section{Ultrastrongly coupled spin-boson models}
\label{sec:spinboson}

The main result of Section \ref{sec:finegraining}, 
Theorem \ref{thm:2}, describes the evolution of the system density matrix
\begin{equation}
\label{scdmat'}
\varrho_\s(\tau) = \llim {\rm tr}_\r\big(e^{-it(H_0+\lambda G\otimes X)}\rho_\s\otimes\rho_\r e^{it(H_0+\lambda G\otimes X)}\big)
\end{equation}
where $H_0=H_\s+H_\r$, {\it under the condition \eqref{n10'}},
\begin{equation}
\label{n10''}
\llim e^{-it(H_0+\lambda G\otimes X)}  =  e^{-i\tau G\otimes X}.
\end{equation}
As we have mentioned, \eqref{n10''} is easily shown to be correct if $H_\r$ and $X$ are bounded operators. We now consider the $\s\r$ model introduced in Section \ref{sec:zenomeasuproc}, with the Hamiltonian \eqref{1} and $X=\varphi(g)$, \eqref{3}. Both these two operators are unbounded. Our main result here is Theorem \ref{thm3-n}, which demonstrates that the expression for $\varrho_\s(\tau)$ given in Theorem \ref{thm:2} is correct for this model.

In order to state a regularity condition on the state $\omega_\r$ and the form factor $g$ used to derive our result (which guarantees \eqref{n10''} to hold in a weak sense, see the proof of Theorem \ref{thm3-n}), we start by noticing that $\forall \tau\ge 0$,
$$
\frac{e^{i\tau\omega(k)/\lambda}-1}{i\omega(k)/\lambda} g(k)\rightarrow \tau g(k)
$$
in the $L^2(\rx^3,d^3k)$ sense as $\lambda\rightarrow\infty$.
We then make the following mild regularity assumption:
\begin{itemize}
\item[{\bf (R)}] For all $\tau\ge 0$,
\begin{equation}
\label{55}
\lim_{\lambda\rightarrow\infty} \omega_\r\Big(W\big(\frac{e^{i\tau\omega/\lambda}-1}{i\omega/\lambda} g\big)\Big) = \omega_\r\big(W(\tau g)\big).
\end{equation}
\end{itemize}

{\bf Example. } For a Gaussian $\omega_\r$, \eqref{9} the condition \eqref{55} holds if 
\begin{equation}
\label{conv}
\sqrt\mathcal C \, \frac{e^{i\tau\omega/\lambda}-1}{i\omega/\lambda} g\ \rightarrow \ \sqrt \mathcal C\tau g
\end{equation}
in the $L^2(\rx^3,d^3k)$ sense, as $\lambda\rightarrow\infty$. If the covariance operator is a multiplication operator by a function $C(k)$ then \eqref{conv} holds if $\sqrt{C(k)} g(k)\in L^2 (\rx^3,d^3k)$. For a thermal reservoir state at inverse temperature $\beta>0$, we have $C(k) = \coth(\beta\omega(k)/2)$.

\begin{thm}
\label{thm3-n}
Let $H$ be the spin-Boson Hamiltonian \eqref{1} and consider the system density matrix \eqref{scdmat'}. Assume the condition (R) and let $D_{l,r}(\tau)=\omega_\r(W(\tau(\gamma_r-\gamma_l)g))$. Then for all $\tau\ge 0$,
\begin{equation}
\label{thm4result-n}
\varrho_\s(\tau) = \sum_{l,r=1}^\nu D_{l,r}(\tau) P_l \rho_\s P_r.
\end{equation}
\end{thm}
We present the proof of Theorem \ref{thm3-n} in Section \ref{sec:thm4proof}. 
As in Theorem \ref{thm:2}, the dynamics \eqref{thm4result-n} can be written in the form \eqref{pop1}, \eqref{pop2} under the appropriate continuity and integrability condition on
$$
D(s)=\omega_\r(e^{-isX}) = \omega_\r(W(-sg)) = \int_{\mathbb R} d\mu(\alpha) e^{-i\alpha s},
$$ 
where $W$ is the Weyl operator, see \eqref{6} and $\mu$ is the associated Borel probability measure, see \eqref{bb31}. 
\medskip

{\bf Examples. } 
\begin{itemize}
    \item[(i)] For $\omega_\r$ a (centered) Gaussian with covariance $\mathcal C$ we have
\begin{equation}
\omega_\r(W(g)) = e^{-\frac14 \langle g, \mathcal C g\rangle}.
\end{equation}
Then $D(s)=e^{-\frac14 s^2\|\sqrt\mathcal C g\|^2}$ is a Gaussian function and $d\mu(\alpha) = \mu(\alpha) d\alpha$ with
$$
\mu(\alpha) = \frac{1}{\|\sqrt\mathcal C g\|\sqrt\pi } e^{-\frac{\alpha^2}{{\|\sqrt \mathcal C}g\|^2}}.
$$

\item[(ii)] A particular case of (i) is when $\omega_\r$ the equilibrium state at inverse temperature $\beta$. The covariance operator is the multiplication by the function $C(\omega) = \coth(\beta\omega/2)$, so 
$$
\|\sqrt{\mathcal C}g\|^2 = \int_{\mathbb R^3} |g(k)|^2 \coth(\beta\omega/2) d^3k.
$$
The zero temperature case (vacuum) corresponds to $\mathcal C=\bbbone$. 
If instead of the thermal distribution the reservoir has an arbitrary energy distribution $\varrho(k)\ge 0$, that is,
$$
\omega_\r\big(a^\dag(k)a(l)\big) = \varrho(k) \delta(k-l),
$$
then the covariance operator is the multiplication with the function $C(k) = 1+2\varrho(k)$. (The thermal distribution is $\varrho(k)=(e^{\beta \omega}-1)^{-1}$.) 

\item[(iii)] A coherent state is given by $W(f)|\Omega\rangle$, where $f\in L^2({\mathbb R^3}, d^3k)$ is fixed and $|\Omega\rangle$ is the vacuum vector. Then 
$$
D(s)=\omega_\r(W(-sg)) = \langle \Omega|W(-f) W(-sg)W(f)|\Omega\rangle = e^{-i s\,{\rm Im}\langle f,g\rangle} e^{-\frac14 s^2 \|g\|^2}.
$$
The corresponding density $\mu$, \eqref{mm39}, becomes
\begin{equation}
\mu(\alpha) = \frac{1}{\sqrt\pi \| g\| } \exp\left(- \frac{\big(\alpha - {\rm Im}\langle f,g\rangle \big)^2 }{ \|g\|^2 }\right).
\end{equation}
\end{itemize}

\subsection{Proof of Theorem \ref{thm3-n}} 
\label{sec:thm4proof}

Define the operator 
\begin{equation}
K =  H_\r +\lambda G\otimes\varphi(g).
\label{HK}
\end{equation}
We have for any system observable $A$,
\begin{equation}
\label{068}
\omega_\s\otimes\omega_\r(e^{it K} Ae^{-it K}) 
= \sum_{l,r} \omega_\s(P_l AP_r) \omega_\r(e^{it (H_\r+\lambda\gamma_l\varphi(g))} e^{-it(H_\r+\lambda\gamma_r\varphi(g))}\big).
\end{equation}
The reservoir observable appearing in the state $\omega_\r$ can be simplified by using the following formula, valid for all $x,y\in\mathbb R$,
\begin{equation}
\label{68}
e^{it (H_\r+x\varphi(g))} e^{-it (H_\r+y\varphi(g))} =e^{-\frac i2 (x^2-y^2){\rm Im}\langle\frac{e^{i\omega t}-1-i\omega t}{\omega^2}g,g\rangle} W\Big((x-y)\frac{e^{i\omega t}-1}{i\omega}g\Big).
\end{equation}
This equality can be derived using the polaron transformation, see for instance \cite{Marcantoni-Merkli}, Lemma 1. Using \eqref{68} in \eqref{068} gives
\begin{equation}
\omega_\s\otimes\omega_\r(e^{it K} Ae^{-it K}) 
= \sum_{l,r=1}^\nu \omega_\s(P_l A P_r) \, e^{-\frac i2 \lambda^2(\gamma^2_l-\gamma_r^2) {\rm Im} \langle (\frac{e^{i\omega t} -1 -i\omega t}{\omega^2}g,g\rangle}  \omega_\r(W\big(\lambda(\gamma_l-\gamma_r)\frac{e^{i\omega t}-1}{i\omega} g\big)).
\label{bh1}
\end{equation}
Next with $\lambda t=\tau$, we get
\begin{equation}
\lambda^2 {\rm Im} \langle \frac{e^{i\omega t} -1 -i\omega t}{\omega^2}g,g\rangle  = \tau^2 {\rm Im} \langle \frac{e^{i\omega t} -1 -i\omega t}{(\omega t)^2}g,g\rangle. 
\end{equation}
For every $\omega>0$ we have $\lim_{t\rightarrow 0}{\rm Im} \frac{e^{i\omega t} -1 -i\omega t}{(\omega t)^2} =0$ and we have $\sup_{x>0} |\frac{e^{ix} -1 -ix}{x^2}|<\infty$. Therefore, by the Lebesgue Dominated Convergence Theorem,
\begin{equation}
\label{bh2}
\llim \lambda^2 {\rm Im} \langle \frac{e^{i\omega t} -1 -i\omega t}{\omega^2}g,g\rangle  = 0.
\end{equation}
The regularity assumption \eqref{55} together with the equations \eqref{bh1}, \eqref{bh2} 
show that 
\begin{equation}
\label{bh5}
\llim \omega_\s\otimes\omega_\r(e^{it K} Ae^{-it K}) = \sum_{l,r=1}^\nu \omega_\s(P_l A P_r)  \omega_\r\big(W(\tau(\gamma_l-\gamma_r) g)\big).
\end{equation}
Note that the right hand side of \eqref{bh5} equals
\begin{eqnarray}
\lefteqn{
\sum_{l,r=1}^\nu \omega_\s\otimes\omega_\r\big( W(\tau\gamma_lg)\, P_lAP_r\, W(-\tau\gamma_rg)\big)}\nonumber\\
&=&\sum_{l,r=1}^\nu \omega_\s\otimes\omega_\r\big( e^{i\tau \gamma_l\varphi(g)}\, P_lAP_r\,  e^{-i\tau \gamma_r\varphi(g)}\big)
\nonumber\\
&=&\sum_{l,r=1}^\nu \omega_\s\otimes\omega_\r\big( e^{i\tau G\otimes\varphi(g)}\, P_lAP_r\,  e^{-i\tau G\otimes \varphi(g)}\big)\nonumber\\
&=& \omega_\s\otimes\omega_\r\big( e^{i\tau G\otimes\varphi(g)} A e^{-i\tau G\otimes \varphi(g)}\big).
\end{eqnarray}
So \eqref{bh5} shows that 
\begin{equation} 
\llim \omega_\s\otimes\omega_\r(e^{it K} (A\otimes\bbbone_\r)e^{-it K}) = \omega_\s\otimes\omega_\r\big( e^{i\tau G\otimes\varphi(g)} (A\otimes\bbbone_\r) e^{-i\tau G\otimes \varphi(g)}\big).
\label{76}
\end{equation}
This means that in the considered limit, $H_\r$ does not play any role in the dynamics of system observables generated by $K$, as only the interaction term in \eqref{HK} survives. 
\medskip

Next we deal with the dynamics when $H_\s$ does not vanish. The total Hamiltonian is given by $H=H_\s+H_\r+\lambda G\otimes\varphi(g)$. We have
\begin{equation}
e^{-it H} - e^{-itK} = \int_0^t e^{-isH}(-iH_\s) e^{-i(t-s)K}ds,
\end{equation}
which implies
\begin{equation}
\label{HH}
\big\| e^{-it H} - e^{-itK}\big\| \le t \|H_\s\|.
\end{equation}
Therefore, in the limit $t\rightarrow 0$, $\lambda\rightarrow\infty$, $\lambda t=\tau$ fixed, the system Hamiltonian does not play a role, 
\begin{equation}
\llim \omega_\s\otimes\omega_\r(e^{it H} Ae^{-it H}) =  \llim \omega_\s\otimes\omega_\r(e^{it K} Ae^{-it K}).
\label{bh6}
\end{equation}
The right side of \eqref{bh6} is given by \eqref{bh5}, which is equivalently expressed as ${\rm tr}(\varrho_\s(\tau) A)$ with $\varrho_\s(\tau)$ given in \eqref{thm4result-n}. This concludes the proof of Theorem \ref{thm3-n}.\hfill \qed

\section*{Acknowledgements} S.M.~received financial support under the Horizon Europe research and innovation programme through the MSCA project ConNEqtions, n.~101056638, and the ERC StG MaTCh, grant agreement n.~101117299. S.M.~also gratefully acknowledges funding from the Italian Ministry of University and Research and Next Generation EU through the PRIN 2022 project ONES, CUP:D53C24003430001. The work of S.M. was performed under the auspices of GNFM-INDAM. M.M.~acknowldges the support from a Discovery Grant of NSERC (Natural Sciences and Engineering Research Council of Canada) as well as the support and hospitality from the Laboratoire J.~A.~Dieudonn\'e at the Universit\'e C\^ote D'Azur, where this work was started.

\end{document}